\documentclass[a4paper]{article}

\usepackage{amsmath,amssymb,amsthm}
\usepackage[margin=1in]{geometry}

\theoremstyle{definition}
\newtheorem{definition}{Definition}
\theoremstyle{plain}
\newtheorem{theorem}[definition]{Theorem}
\newtheorem{lemma}[definition]{Lemma}
\newtheorem*{corollary}{Corollary}
\theoremstyle{remark}
\newtheorem*{example}{Example}

\usepackage{microtype}

\usepackage{hyperref}

\hypersetup{colorlinks=true,linkcolor=blue,citecolor=blue,urlcolor=blue}

\allowdisplaybreaks[4]

\renewcommand\emptyset{\varnothing}
\newcommand\fin{_{\mathrm f}}
\newenvironment{itemize*}{\renewcommand\item{\ignorespaces}\ignorespaces}{}

\usepackage{pb-diagram}

\title{Systematic Construction of Temporal Logics for Dynamical Systems via Coalgebra}

\author{Baltasar Tranc{\'o}n y Widemann\\[1ex]
  Computer Science, Technical University of Ilmenau\\
  Bayreuth, DE\\
  \texttt{Baltasar.Trancon@tu-ilmenau.de}}

\begin{document}

\maketitle

\begin{abstract}
  Temporal logics are an obvious high-level descriptive companion
  formalism to dynamical systems which model behavior as deterministic
  evolution of state over time.  A wide variety of distinct temporal
  logics applicable to dynamical systems exists, and each candidate
  has its own pragmatic justification.  Here, a systematic approach to
  the construction of temporal logics for dynamical systems is
  proposed: Firstly, it is noted that dynamical systems can be seen as
  coalgebras in various ways.  Secondly, a straightforward standard
  construction of modal logics out of coalgebras, namely Moss's
  coalgebraic logic, is applied. Lastly, the resulting systems are
  characterized with respect to the temporal properties they express.
\end{abstract}

\section{Introduction}

\emph{Dynamical systems} are the classical constructive formalism for
behaviour arising from the deterministic evolution of system state
over time \cite{Birkhoff1927}, dating back to the works of Newton and
Laplace.  Clearly \emph{temporal logics}, with operators such as
`next', `always', `eventually' and `for-at-least', constitute a
companion descriptive formalism.  However, the relation is not
one-to-one: One one hand, there is a unifying theory underlying the
various perspectives on dynamical systems as monoid actions, which
uniformly covers discrete and continuous, as well as hybrid systems
\cite{Jacobs2000}.  But on the other hand, the diversity of temporal
logics in literature is immense, cf.~\cite{Venema2001}, and the choice
for a particular system is often justified by ad-hoc pragmatic
arguments.  The present article explores a systematic and fairly
generic approach to the construction of temporal logics for dynamical
systems, via the rather recent mathematical field of \emph{universal
  coalgebra} which appears to be intimately connected to both
dynamical systems \cite{Rutten2000} and modal logics \cite{Kurz2008}.
A different approach also based on coalgebras and the Stone duality
has been suggested \cite{Bonsangue2005} for constructing modal logics
of \emph{transition systems}, a close relative of dynamical systems in
computer science.

The method outlined in the remainder of this article, while
theoretically simple, touches on many different fields of mathematics:
order theory, category theory, algebra, coalgebra, classical modal
logics à la Kripke, and coalgebraic logics à la Moss \cite{Moss1999}.
Thus a significant proportion of the available space is dedicated to
reviewing the relevant definitions and propositions from the
respective standard literature.  This review makes up the sections
\ref{ingredients1} and \ref{ingredients2}.  The expert reader is
encouraged to skip ahead: Section~\ref{constructions} ties up all the
loose ends and gives a novel contribution.  There a selection of
obvious coalgebraic perspectives on dynamical systems is explored, and
the respective logics entailed by applying Moss's construction are
characterized.

\section{Review: Classical Ingredients}
\label{ingredients1}

This section reviews some basic definitions and propositions.

\subsection{Order Relations}

We assume that the reader is familiar with basic order-theoretic
properties of binary relations, namely with \emph{reflexive},
\emph{transitive}, \emph{symmetric} relations, and with
\emph{preorders}, \emph{partial orders} and \emph{equivalences}.  We
give two additional related definitions that are not quite as
universal:

\begin{definition}
  Let $X$ be a set.  A binary relation $R \subseteq X^2$ is called
  \begin{itemize*}
  \item \emph{non-branching} if and only if $x \mathrel R y$ and $x
    \mathrel R z$ imply $y \mathrel R z$ or $z \mathrel R y$, and
  \item \emph{linear} if and only if $x \mathrel R y$ or $y \mathrel R
    x$,
  \end{itemize*}
  respectively, for all $x, y, z \in X$.
\end{definition}

\subsection{Monoids}

We assume that the reader is familiar with the notions of a
\emph{monoid} $\mathbb M = (M, 0, +)$, and of monoid \emph{generators}
and \emph{cyclic} monoids.  Every monoid induces an ordering relation.

\begin{definition}[Monoid Order]
  Let $\mathbb M = (M, 0, {+})$ be a monoid.  For any elements $a, b
  \in M$, we write $a \leq_{\mathbb M} b$ if and only if there is some
  $c \in M$ such that $a + c = b$.  We say that $a \leq_{\mathbb M} b$
  \emph{via} $c$.  It follows directly from the monoid axioms that
  $\leq_{\mathbb M}$ is reflexive and transitive, hence a preorder.
  By extension, $\mathbb M$ itself is called
  \emph{symmetric}/\emph{non-branching}/\emph{linear} if and only if
  $\leq_{\mathbb M}$ is symmetric/non-branching/linear, respectively.
\end{definition}

Note that being symmetric in this sense is different from being
Abelian.  In fact, symmetry characterizes a subclass of monoids, the
groups.

\begin{lemma}[Groups]
  A monoid $\mathbb M$ is a group if and only if it is symmetric.
  Every symmetric monoid is trivially linear, with the degenerate
  order $(\leq_{\mathbb M}) = M^2$, the full relation.
\end{lemma}




\subsection{Dynamical Systems}

\begin{definition}[Dynamical System]
  Let $\mathbb T = (T, 0, {+})$ be a monoid called \emph{time}.  A
  \emph{dynamical system} is an enriched structure $\mathbb S =
  (\mathbb T, S, \Phi)$ with
  \begin{itemize*}
  \item a set $S$ called \emph{state space}, and
  \item a map $\Phi : S \times T \to S$ called \emph{dynamics},
  \end{itemize*}
  such that
  \begin{align}
    \Phi(s, 0) &= s & \Phi\bigl(\Phi(s, t), u\bigr) &= \Phi(s, t + u)
  \end{align}
  In other words, $\Phi$ is a \emph{right monoid action} of $\mathbb
  T$ on $S$.  $\mathbb S$ is called
  \begin{itemize*}
  \item \emph{linear-time} if and only if $\mathbb T$ is linear,
    otherwise \emph{nonlinear-time}, and
  \item \emph{invertible} if and only if $\mathbb T$ is symmetric.
  \end{itemize*}
\end{definition}

\begin{corollary}
  There are no invertible nonlinear-time dynamical systems.
\end{corollary}

Dynamical systems are a fundamental model class of many natural and
social sciences.  In comparison with their younger counterpart in
computer science, automata and transition systems, dynamical systems
are typically
\begin{itemize}
\item behaviourally weaker -- deterministic, non-pointed (without
  distinguished initial states) and total (without spontaneous
  termination), but
\item structurally stronger -- with additional features of time
  (density, completeness) and state space (topology, metric,
  differential geometry, measures).
\end{itemize}
Automata-like construction can be emulated by dynamical systems; see
examples below.

\begin{definition}[Step, Trajectory, Orbit]
  From the dynamics map we may derive three forms of auxiliary functions:
  \begin{align*}
    \Phi^t &: S \to S
    &
    \Phi_s &: T \to S
    &
    \Phi^\circ &: S \to \mathcal PS
    \\
    \Phi^t(s) &= \Phi(s, t) &
    \Phi_s(t) &= \Phi(s, t) &
    \Phi^\circ(s) &= \mathrm{Img}(\Phi_s) = \{ \Phi(s, t) \mid t \in T \}
  \end{align*}
  \begin{itemize*}
  \item $\Phi^t$ is called the \emph{step} of \emph{duration} $t$, or
    just the $t$-step.
  \item $\Phi_s$ is called the \emph{trajectory} of \emph{initial}
    state $s$.
  \item $\Phi^\circ(s)$ is called the \emph{orbit} of state $s$.
  \end{itemize*}
\end{definition}

\begin{lemma}[Homomorphic Steps]
  The dynamical systems with time $\mathbb T$ are precisely those
  systems $(\mathbb T, S, \Phi)$ such that the step construction is a
  monoid homomorphism from $\mathbb T$ into the monoid of functions of
  type $S \to S$ with right composition.
  \begin{align}
    \Phi^0 &= \mathrm{id}_S & \Phi^{t+u} &= \Phi^u \circ \Phi^t
  \end{align}
  where $\mathrm{id}_X(x) = x$ and $(f \circ g)(x) = f(g(x))$ for all $x$.
\end{lemma}

\begin{corollary}[Generating Steps]
  If $G \subseteq T$ is a generator of $\mathbb T$, then $\Phi$ is
  determined uniquely by the collection of steps $\{ \Phi^t \mid t \in
  G \}$.
\end{corollary}

\begin{example}[Instances of Time]
  \label{ex:time}%
  \leavevmode
  \begin{itemize}
  \item The time monoid $(\mathbb{N}, 0, +)$ yields standard
    non-invertible, discrete-time dynamical systems.  The step
    $\Phi^1$ is generating.  Trajectories are (one-sided) infinite
    sequences.
  \item The time monoid $(\mathbb{Z}, 0, +)$ yields standard
    invertible, discrete-time dynamical systems.  The step $\Phi^1$ is
    generating and must be invertible.  Trajectories are two-sided
    infinite sequences.
  \item The time monoid $(\mathbb{R}_+, 0, +)$ yields standard
    non-invertible, continuous-time dynamical systems.  No simple step
    generator exists.  Trajectories are one-sided parametric curves.
  \item The time monoid $(\mathbb{R}, 0, +)$ yields standard
    invertible, continuous-time dynamical systems.  No simple step
    generator exists; classical definitions are given as solutions to
    ordinary differential equations.  Trajectories are two-sided
    parametric curves.
  \item The ``time'' monoid $(\Sigma^*, \varepsilon, \cdot)$ over some
    finite alphabet $\Sigma$ yields total semiautomata, or
    deterministic finitely-labelled transition systems.  The steps $\{
    \Phi^a \mid a \in \Sigma \}$ (columns of the transition table) are
    generating.  Trajectories are big-step transition functions of
    total automata, mapping input words to final states.
  \end{itemize}
\end{example}

\subsection{Propositional Modal Logics}

We assume that the reader is familiar with the syntax and semantics of
classical propositional logics and their presentation in terms of the
connectives $\neg$ and $\to$.

\begin{definition}[Syntax of Propositional Modal Logics]
  The modal extension of classical propositional logics adds two unary
  connectives $\Box$ and $\lozenge$, taking $\Box$ as primitive and defining
  \begin{equation*}
    \lozenge A = \neg \Box \neg A
  \end{equation*}
\end{definition}

\begin{definition}[Semantics of Propositional Modal Logics]
  A \emph{normal} modal extension of classical propositional logics
  adds at least the deduction rule of \emph{necessitation} or
  \emph{generalization}, and the axiom of \emph{distribution}:
  \begin{align*}
    A \vdash \Box A &&
    \Box(A \to B) \to (\Box A \to \Box B)
  \end{align*}
\end{definition}

\begin{example}
  Important normal modal logics are obtained by adding certain axioms:
  \begin{itemize}
  \item $\Box A \to A$ added to the minimal system results in the logic $T$.
  \item $\Box A \to \Box \Box A$ added to $T$ results in the logic
    $S4$.
  \item $\Box(\Box A \to B) \lor \Box(\Box B \to A)$ added to $S4$
    results in the logic $S4.3$.
  \item $\lozenge A \to \Box \lozenge A$ added to $S4$ or $S4.3$
    results in the logic $S5$.
  \end{itemize}
\end{example}

\subsection{Kripke Semantics}

\begin{definition}[Kripke Frame]
  A Kripke frame is a structure $(W, R)$ with a set $W$ of
  \emph{worlds} and a relation $R$ on $W$ called \emph{accessibility}.
\end{definition}

\begin{definition}[Kripke Model]
  Let $(W, R)$ be a Kripke frame.  A Kripke model (of propositional
  modal logic) is an extended structure $(W, R, \Vdash)$, where
  $\Vdash$ is a relation between $W$ and the language $\mathit{Prop}$
  of logical formulas, such that
  \begin{equation*}
    \begin{aligned}
      w &\Vdash \neg A && \iff && w \not\Vdash A
      \\
      w &\Vdash A \to B && \iff && w \not\Vdash A \text{~or~} w \Vdash B
      \\
      w &\Vdash \Box A && \iff && v \Vdash A \text{~whenever~} w \mathrel{R} v
    \end{aligned}
  \end{equation*}
  \nopagebreak
  We say that $w$ \emph{satisfies} $A$ in $(W, R, \Vdash)$ if and only
  if $w \Vdash A$.
\end{definition}

\begin{lemma}
  The satisfaction relation $\Vdash$ of a Kripke frame is determined
  uniquely by the satisfaction of atomic propositions.
\end{lemma}

\begin{definition}[Validity]
  A formula $A$ is called \emph{valid} in
  \begin{itemize}
  \item a world $w$ if and only if $w$ satisfies $A$,
  \item a Kripke model $(W, R, {\Vdash})$ if and only if it is valid
    in all worlds $w \in W$,
  \item a Kripke frame $(W, R)$ if and only if it is valid in all
    Kripke models $(W, R, \Vdash)$,
  \item a class $C$ of Kripke frames if and only if it is valid in all
    members of $C$.
  \end{itemize}
\end{definition}

\begin{definition}[Soundness/Completeness]
  A propositional modal logic $L$ is called, with respect to a class
  $C$ of Kripke frames,
  \begin{itemize*}
  \item \emph{sound} if and only if truth in $L$ implies validity in
    $C$, and
  \item \emph{complete} if and only if validity in $C$ implies truth
    in $L$.
  \end{itemize*}
\end{definition}

\begin{theorem}[Soundness/Completeness]
  \label{theorem:log-frame}
  The modal logics $S4$/$S4.3$/$S5$ are sound and complete for the
  class of Kripke frames $(W, R)$ where $R$ is an
  arbitrary/non-branching/symmetric preorder, respectively.
\end{theorem}

\begin{definition}[Finite Model Property]
  A propositional modal logic $L$ is said to have the \emph{finite
    model property}, if and only if it is complete for a class of
  finite Kripke frames.
\end{definition}

\begin{theorem}
  \label{theorem:log-fmp}
  The modal logics $S4$/$S4.3$/$S5$ have the finite model property,
  for subclasses of the respective classes given in
  Theorem~\ref{theorem:log-frame}.
\end{theorem}

\section{Review: Additional Ingredients}
\label{ingredients2}

This section reviews some definitions and propositions that are also
basic, but from less well-known fields.  See
\cite{Rutten2000,Moss1999} for greater detail.

\subsection{Category Theory}

\begin{definition}[Set Endofunctor]
  A \emph{functor} $F$ on the category of sets, or \emph{set
    endofunctor}, is a map that assigns
  \begin{itemize*}
  \item to every set $X$ a set $FX$, and
  \item to every function $h : X \to Y$ a function $Fh : FX \to FY$,
  \end{itemize*}
  such that
  \begin{align*}
    F(\mathrm{id}_X) &= \mathrm{id}_{FX} & F(g \circ h) &= Fg \circ Fh
  \end{align*}
  where $\mathrm{id}_X(x) = x$ and $(g \circ f)(x) = g\bigl(f(x)\bigr)$.
\end{definition}

\noindent All functors considered in the following are tacitly set
endofunctors.

\begin{definition}[Monotonic Functor]
  A functor $F$ is called \emph{monotonic} if and only if $X \subseteq
  Y$ implies $FX \subseteq FY$.
\end{definition}

Coalgebraic logics deal with a class of functors called
\emph{standard}, which are essentially monotonic, plus an additional
condition, namely preservation of weak pullbacks, that is rather
technical but fortunately inessential for the present discussion.

\begin{definition}[Finitary Functor]
  A functor is called \emph{finitary} if and only if
  \begin{equation*}
    FX \subseteq \textstyle\bigcup \{ FY \mid Y \subseteq X; Y \text{~finite} \}
  \end{equation*}
  otherwise \emph{infinitary}.  For monotonic finitary functors, the above
  is necessarily an equality.  A standard, infinitary functor $F$
  has a \emph{finitary restriction} $F\fin$ defined by
  \begin{align*}
    F\fin X &= \textstyle\bigcup \{ FY \mid Y \subseteq X; Y \text{~finite} \}
    &
    F\fin (h : X \to Y) &= F h \rvert_{F\fin X}
  \end{align*}
\end{definition}

\begin{definition}[Functor Product]
  The pointwise Cartesian product of functors is again a functor.
  \begin{align*}
    (F \times G)X &= FX \times GX & \bigl((F \times G)h\bigr)(x, y) &= \bigl((Fh)(x), (Gh)(y)\bigr)
  \end{align*}
\end{definition}

\begin{example}
  The following are standard functors:
  \begin{itemize}
  \item The \emph{identical} functor $\mathcal I$
    \begin{align*}
      \mathcal IX &= X & \mathcal Ih &= h
    \end{align*}
    $\mathcal I$ is finitary; hence $\mathcal I\fin = \mathcal I$.
  \item The \emph{constant} functor $\_ @ C$ for some set $C$
    \begin{align*}
      X @ C &= C & h @ C &= \mathrm{id}_C
    \end{align*}
    $\_ @ C$ is finitary.
  \item the \emph{powerset} functor $\mathcal P$
    \begin{align*}
      \mathcal PX &= \{ W \mid W \subseteq X \} & (\mathcal Ph)(W) &=
      \{ h(x) \mid x \in W \}
    \end{align*}
    $\mathcal P$ is not finitary; its finitary restriction is the
    \emph{finite powerset} functor $\mathcal P\fin$.
  \item the \emph{Hom} functor $\_^C$ for some set $C$
    \begin{align*}
      X^C &= \{ f \mid f : C \to X \} & (h^C)(g) &= h \circ g
    \end{align*}
    $\_^C$ is finitary if and only if $C$ is finite; its finitary
    restriction is the \emph{image-finite} functor $\_^C\fin$.
  \end{itemize}
\end{example}

Clearly, a relation $R \in \mathcal P(X \times Y)$ is precisely the
set of pairs $(x, y)$ for which there is some $r \in R$ such that
$\pi_1(r) = x$ and $\pi_2(r) = y$, where $\pi_1, \pi_2$ are the
natural projections from the binary Cartesian product.  This seemingly
redundant presentation suggests an interaction of relations and
functors.

\begin{definition}[Relational Lifting]
  Let $F$ be a functor.  Every relation $R \in \mathcal P(X \times Y)$
  has a \emph{lifting} $F[R] \in \mathcal P(FX \times FY)$ defined as
  the set of pairs $(\hat x, \hat y)$ for which there is some $\hat r
  \in FR$ such that $(F \pi_1)(\hat r) = \hat x$ and $(F \pi_2)(\hat
  r) = \hat y$.
\end{definition}

\begin{example}
  The liftings for the functors discussed above are as follows:
  \begin{itemize}
  \item The identical functor lift a relation to itself: $x
    \mathrel{\mathcal I[R]} y$ if and only if $x \mathrel R y$.
  \item The constant functor lifts to the identity relation: $c
    \mathrel{[R] @ C} c'$ if and only if $c = c' \in
    C$.
  \item $Y \mathrel{\mathcal P[R]} Z$ if and only if for all $y \in Y$
    there is a $z \in Z$, and vice versa, such that $y \mathrel{R} z$.
  \item $f \mathrel{[R]^C} g$ if and only if $f(c) \mathrel R g(c)$
    for all $c \in C$.
  \end{itemize}
\end{example}

\subsection{Universal Coalgebra}

\begin{definition}[Coalgebra]
  Let $F$ be a functor.  An $F$-\emph{coalgebra} is a structure $(X,
  f)$ with an object $X$ and an arrow $f : X \to FX$.
\end{definition}

\begin{definition}[Homomorphism]
  Let $F$ be a functor.  Let $(X, f)$ and $(Y, g)$ be $F$-coalgebras.
  An $F$-\emph{coalgebra homomorphism} from $(X, f)$ to $(Y, g)$ is an
  arrow $h : X \to Y$ such that $Fh \circ f = g \circ h$.  We write $h
  : (X, f) \to (Y, g)$ or simply $h : f \to g$.
\end{definition}

\begin{definition}[Final Coalgebra]
  Let $F$ be a functor.  An $F$-coalgebra $(Z, z)$ is called
  \emph{final} if and only if there is a unique homomorphism $f! : f
  \to z$ from any other $F$-coalgebra.
\end{definition}

\begin{theorem}
  Every finitary functor has a final coalgebra.
\end{theorem}

\begin{definition}[Bisimulation]
  Let $F$ be a functor.  Let $(X, f)$ and $(Y, g)$ be $F$-coalgebras.
  A \emph{bisimulation} between $(X, f)$ and $(Y, g)$ is a relation $R
  \subseteq X \times Y$ that can be extended to an $F$-coalgebra $(R,
  r)$ such that the projections are coalgebra homomorphisms $\pi_1 : r
  \to f$ and $\pi_2 : r \to g$.  We say that states $x \in X$ and $y
  \in Y$ are \emph{bisimilar} if and only if there is a bisimulation
  relating them.
\end{definition}

The final coalgebra can be seen as a system of representatives of
equivalence classes modulo bisimilarity.

\begin{theorem}
  Let $F$ be a standard functor. If a final $F$-coalgebra $(Z, z)$
  exists then, for given $F$-coalgebras $(X, f)$ and $(Y, g)$, two
  states $x \in X; y \in Y$ are bisimilar if and only if $f!(x) =
  g!(y)$.
\end{theorem}

\begin{definition}[Parallel Coalgebra Composition]
  \label{def:coalg-par}%
  Coalgebras with the same carrier can be combined in parallel: Let
  $(X, f)$ be an $F$-coalgebra and $(X, g)$ be a $G$-coalgebra.  Then
  $(X, \langle f, g \rangle)$ is an $(F \times G)$-coalgebra, where
  \begin{equation*}
    \langle f, g \rangle(x) = \bigl(f(x), g(x)\bigr)
  \end{equation*}
\end{definition}

\subsection{Moss's Coalgebraic Logic}

The idea of Moss's coalgebraic logic \cite{Moss1999} is to replace
Kripe frames by $F$-coalgebras for some functor $F$, and to derive a
universal and natural modality from $F$ itself.

\begin{definition}[Moss's Coalgebraic Logic, Abstract]
  Fix a standard functor $F$.  Extend the syntax of propositional
  logic by a pseudo-unary connective $\nabla$ that, unlike the
  classical modalities like $\Box$, applies not to a single formula $A
  \in \mathit{Prop}$ but to an expression of type either $\widehat A
  \in F(\mathit{Prop})$ or $\widehat A \in F\fin(\mathit{Prop})$.  For
  infinitary $F$ where the choice makes a difference, the cases are
  called \emph{infinitary} and \emph{finitary} $F$-coalgebraic logics,
  respectively.  A Moss model is a structure $(X, f, \Vdash)$ where
  $(X, f)$ is an $F$-coalgebra and $\Vdash$ is a relation between
  coalgebra states and formulas, such that
  \begin{align*}
    x \Vdash \neg A \iff x \not\Vdash A
    &&
    x \Vdash A \to B \iff x \not\Vdash A \text{~or~} x \Vdash B
  \end{align*}
  as for Kripke models, but
  \begin{equation*}
    x \Vdash \nabla \widehat A \iff f(x) \mathrel{F[\Vdash]} \widehat A
  \end{equation*}
\end{definition}

Moss's coalgebraic logic as presented here specifies satisfaction only
up to atomic propositions, in analogy to Kripke frames.  In Moss's
original presentation, the specification is unique, in analogy to
Kripke models.

\begin{definition}[Moss's Coalgebraic Logic, Concrete]
  Let $(X, f)$ be an $F$-coalgebra.  Let $s : X \to \mathcal
  P(\mathit{Prop}_0)$ be the map that assigns to each state $x \in X$
  the desired set of valid atomic propositions.  Then $(X, s)$ is a
  $\mathrm{Const}\bigl(\mathcal P(\mathit{Prop}_0)\bigr)$-coalgebra.
  For the parallel composite coalgebra $(X, g = \langle f, s
  \rangle)$, a unique Moss model is specified by the additional clause
  \begin{align*}
    x \Vdash A \iff A \in s(x) \qquad (A \in \mathit{Prop}_0)
  \end{align*}
\end{definition}

The following two propositions state that traditional Kripke frames
are essentially equivalent to the special case $F = \mathcal P$.

\begin{lemma}
  \label{lemma:coalg-kripke-correspondence}%
  $\mathcal{P}$-coalgebras $(X, f)$ are in one-to-one correspondence
  to relations $R$ on $X$ by putting $x \mathrel{R} y$ if and only if
  $y \in f(x)$.
\end{lemma}

\begin{theorem}
  The Kripke modalities $\Box, \lozenge$ and the Moss modality
  $\nabla$ for finitary $\mathcal P$-coalgebraic logics are
  equivalent.  For infinitary $\mathcal P$-coalgebraic logics, they
  are also equivalent in the presence of infinitary conjunction and
  disjunction; otherwise $\nabla$ is generally more expressive.
  \begin{align*}
    &w \Vdash_{\mathrm K} \Box A \iff w \Vdash_{\mathrm M} \nabla \{ A
    \} \lor \nabla \emptyset & &w \Vdash_{\mathrm K} \lozenge A \iff w
    \Vdash_{\mathrm M} \nabla \{ A, \top \}
    \\
    &w \Vdash_{\mathrm M} \nabla \widehat A \iff w \Vdash_{\mathrm K}
    \Box \left({\textstyle\bigvee \widehat A}\right) \land {\textstyle
      \bigwedge \lozenge \widehat A} & &\text{where}\quad \lozenge
    \widehat A = \{ \lozenge B \mid B \in \widehat A \}
  \end{align*}
  where $\Vdash_{\mathrm K}$/\/$\Vdash_{\mathrm M}$ denote satisfaction
  {\`a} la Kripke/Moss, respectively.
\end{theorem}

In general, the infinitary version of the operator $\nabla$ is better
matched with a logic where conjunction and disjunction are also
infinitary.  While an uncommon topic classically, infinitary logics
are an important topic in modal logic because of their connection to
bisimulation.  The following theorem generalizes a theorem of
Kripke-style logic, where bisimilarity is defined ad-hoc but
equivalently to the coalgebraic notion specialized as in
Lemma~\ref{lemma:coalg-kripke-correspondence}.

\begin{theorem}
  In fully ($\land, \lor, \nabla$) infinitary $F$-coalgebraic logic,
  two states $s, t \in S$ satisfy the same set of formulas if and only
  if they are bisimilar.
\end{theorem}


\section{Constructions}
\label{constructions}

This section gives novel theoretical results by invetigating the
ramifications of the following recipe:
\begin{enumerate}
\item identify some generic $F$-coalgebraic view on dynamical systems,
\item use Moss's construction to obtain logics with $\nabla_F$
  modality, depending on the functor $F$,
\item relate $\nabla_F$ to established temporal logic operators.
\end{enumerate}

Note that all of the following constructions have the state space $S$
of a fixed dynamical system as the carrier of some coalgebra for
various functors.  Hence the associated logical languages can coexist
naturally in a single system, by the parallel composition given in
Definition~\ref{def:coalg-par}.

\subsection{Step Logics}

\begin{definition}[Step Coalgebra]
  Let $\mathbb S = (\mathbb T, S, \Phi)$ be a dynamical system.  For
  any element $t \in T$, the $\mathcal I$-coalgebra $(S, \Phi^t)$ is
  called the $t$-\emph{step coalgebra} of $\mathbb S$.
\end{definition}

\begin{definition}[Multi-Step Coalgebra]
  Let $\mathbb S = (\mathbb T, S, \Phi)$ be a dynamical system.  For
  any subset $U \subseteq T$, the $\_^U$-coalgebra $(S, s \mapsto
  \Phi_s \circ \mathrm{in})$, given the inclusion map $\mathrm{in} : U
  \to T$, is called the $U$-\emph{multi-step coalgebra} of $\mathbb
  S$.
\end{definition}

\begin{lemma}
  The $\nabla$ modality of step coalgebras amounts to
  \begin{itemize}
  \item for the $t$-step:
    \begin{equation*}
      s \Vdash \nabla A  \iff  \Phi(s, t) \Vdash A
    \end{equation*}
  \item for the $U$-multi-step:
    \begin{equation*}
      s \Vdash \nabla \widehat A  \iff  \Phi(s, t) \Vdash \widehat A(t) \text{~for all~} t \in U
    \end{equation*}
  \end{itemize}
  The functors for $t$-steps and finite $U$-multi-steps are finitary;
  hence no additional distinction between finitary and infinitary
  logics arises.
\end{lemma}

\begin{definition}[Step Modality]
  \begin{align*}
    \bigcirc A &= \nabla A &
    \bigcirc_t A &= \nabla u \mapsto
      \begin{cases}
        A & (t = u)
        \\
        \top & (t \neq u)
      \end{cases}
  \end{align*}
\end{definition}

\begin{example}
  \label{ex:steps}
  (Multi-)Step coalgebras are of particular interest for finite
  generators, since they specify the dynamics uniquely and concisely.
  The following are generating, cf.\ Example~\ref{ex:time}:
  \begin{itemize}
  \item For time $(\mathbb N, 0, +)$, the $1$-step coalgebra maps
    every state to its successor.  The resulting temporal logic has
    $\bigcirc$ as the \emph{next} operator of traditional unidirectional
    discrete-time temporal logic.
  \item For time $(\mathbb Z, 0, +)$, the $(\pm 1)$-step coalgebra
    maps every state to its successor/predecessor, respectively.  The
    resulting temporal logic has $\bigcirc_{\pm1}$ as the
    \emph{next}/\emph{previously} operators of traditional
    bidirectional discrete-time temporal logic, respectively.
  \item For ``time'' $(\Sigma^*, \varepsilon, \cdot)$, the
    $\Sigma$-multi-step coalgebra maps every automaton state to its
    response function (row of the transition table).  The resulting
    logic has $(\bigcirc_a)_{a\in\Sigma}$ as the generating cases of
    Pratt's \emph{necessity} operators $[a]$ in dynamic
    logic \cite{Pratt1976}, where they are extended to the free Kleene
    algebra over $\Sigma$.
  \end{itemize}
  Interesting infinite, non-generating examples include:
  \begin{itemize}
  \item For time $(\mathbb R, 0, +)$ and $\delta > 0$, let $U$ denote
    the open interval $(-\delta, \delta)$.  The $U$-multi-step
    coalgebra maps every state to its temporal $\delta$-neighbourhood.
  \end{itemize}
\end{example}

\begin{lemma}
  The modality $\nabla$ and the family of modalities $(\bigcirc_t)_{t
    \in U}$ for generating $U$ are straightforwardly equivalent if $U$
  is finite, and equivalent in the presence of infinitary conjunction
  otherwise.
  \begin{equation*}
    x \Vdash \nabla \widehat A \iff x \Vdash \bigwedge_{t \in U} \bigcirc_t \widehat A(t)
  \end{equation*}
\end{lemma}

The following construction is the multi-step limit case $U = T$.

\subsection{Trajectory Logics}

\begin{definition}[Trajectory Coalgebra]
  Let $\mathbb S = (\mathbb T, S, \Phi)$ be a dynamical system.  The
  $\_^T$-coalgebra $(S, s \mapsto \Phi_s)$ is called the
  \emph{trajectory coalgebra} of $\mathbb S$.
\end{definition}

\begin{lemma}
  The $\nabla$ modality of trajectory coalgebras amounts to
  \begin{equation*}
    s \Vdash \nabla \widehat A  \iff  \Phi(s, t) \Vdash \widehat A(t) \text{~for all~} t \in T
  \end{equation*}
\end{lemma}

The $\nabla$ trajectory modality is a surprisingly powerful logical
operator, with the severe disadvantage that there is no canonical
syntactic representation.  The following examples are but a small
subset of useful special cases.

\begin{example}
  Arguments of the $\nabla$ trajectory modality are functions of type
  $T \to \mathit{Prop}$.  Various intensional notations for such
  functions, or time-dependent formulas, give rise to well-known
  temporal operators.  Note that all following examples work for
  finitary $\nabla$.
  \begin{itemize}
  \item Consider discrete time $(\mathbb N, 0, +)$ or $(\mathbb Z, 0,
    +)$.  Define a \emph{zip} operator as
    \begin{equation*}
      A \leftrightharpoons B = \nabla t \mapsto
      \begin{cases}
        A & t \text{~even}
        \\
        B & t \text{~odd}
      \end{cases}
    \end{equation*}
    Then a dynamic system is bipartite, with characteristic
    formula $A$, if and only if $(A \leftrightharpoons
    \neg A) \lor (\neg A \leftrightharpoons A)$ is valid in the Moss
    model associated with its trajectories.
  \item Consider automaton time $(\Sigma^*, \varepsilon, \cdot)$.  
    Define a \emph{consumption} operator as
    \begin{equation*}
      \mathit{eat}(L, A, B) = \nabla t \mapsto
      \begin{cases}
        A & t \in L
        \\
        B & t \not\in L
      \end{cases}
    \end{equation*}
    for languages $L \subseteq \Sigma^*$ and formulas $A, B$.  Now let
    $A$ be a formula characterizing accepting states.  Then an
    automaton, as a dynamical system, accepts
    \begin{itemize}
    \item at least the language $L \subseteq \Sigma^*$ if and only if
      $\mathit{eat}(L, A, \top)$
    \item exactly the language $L \subseteq \Sigma^*$ if and only if
      $\mathit{eat}(L, A, \neg A)$
    \end{itemize}
    is valid for its initial state(s) in the Moss model associated
    with its trajectories.
  \item Consider time with a linear antisymmetric order $<$.  Define a
    \emph{change} operator as
    \begin{equation*}
      \mathit{chg}(t, A, B, C) = \nabla u\mapsto
      \begin{cases}
        A & u < t
        \\
        B & u = t
        \\
        C & u > t
      \end{cases}
    \end{equation*}
    for time duration $t$ and formulas $A, B, C$.  Then
    minimum/maximum-duration operators can be defined directly, in two
    variants differing in the inclusion of boundary cases:
    \begin{align*}
      \mathrm{min}\, t.\; A &= \mathit{chg}(t, A, \top, \top) &
      \mathrm{max}\, t.\; A &= \mathit{chg}(t, \top, \top, \neg A)
      \\
      \mathrm{min}'\, t.\; A &= \mathit{chg}(t, A, A, \top)
      &
      \mathrm{max}'\, t.\; A &= \mathit{chg}(t, \top, \neg A, \neg A)
    \end{align*}
    Imprecise operators such as \emph{until} can be expressed as
    infinitary disjunctions:
    \begin{equation*}
      A \mathbin{\mathbf{U}} B = \bigvee_{t \in T} \mathit{chg}(t, A, B, \top)
    \end{equation*}
  \end{itemize}
\end{example}

\subsection{Orbit Logics}

The following construction shifts the coalgebraic focus from
trajectories to orbits which are images of trajectories, hence
abstracting from durations.  The result is a family of qualitive
temporal logics that can be expressed naturally in the classical modal
operators, uniformly for all kinds of time structure.

\begin{definition}[Orbit Coalgebra]
  Let $\mathbb S = (\mathbb T, S, \Phi)$ be a dynamical system.  The
  $\mathcal P$-coalgebra $(S, \Phi^\circ)$ is called the \emph{orbit
    coalgebra} of $\mathbb S$.  We say that in $\mathbb S$, $y$ is
  \emph{reachable} from $x$, written $x \leadsto_{\mathbb S} y$, if
  and only if $y \in \Phi^\circ(x)$.
\end{definition}

\begin{lemma}
  \label{lemma:dyn2rel}%
  For dynamical systems $\mathbb S$, the reachability relation
  $\leadsto_{\mathbb S}$ is
  \begin{enumerate}
  \item always a preorder,
  \item additionally non-branching, but not generally linear, if
    $\mathbb S$ is linear-time,
  \item additionally symmetric if $\mathbb S$ is invertible.
  \end{enumerate}
\end{lemma}

\begin{proof}
  We have $x \leadsto_{\mathbb S} y$ if and only if there is some $t$
  such that $\Phi(x, t) = y$.  We say $x \leadsto_{\mathbb S} y$ via
  $t$.
  \begin{enumerate}
  \item Reflexivity and transitivity follow directly from the monoid
    axioms: $x \leadsto_{\mathbb S} x$ via $0$, and if $x
    \leadsto_{\mathbb S} y$ via $t$ and $y \leadsto_{\mathbb S} z$ via
    $u$, then $x \leadsto_{\mathbb S} z$ via $t + u$.
  \item Assume that $x \leadsto_{\mathbb S} y$ via $t$ and $x
    \leadsto_{\mathbb S} z$ via $u$.  By linearity of $\mathbb T$
    assume, without loss of generality, that $t \leq_{\mathbb T} u$
    via $v$.  Then $y \leadsto_{\mathbb S} z$ via $v$.
  \item For symmetric $\mathbb T$, if $x \leadsto_{\mathbb S} y$ via
    $t$, then $y \leadsto_{\mathbb S} x$ via $-t$. \qedhere
  \end{enumerate}
\end{proof}

\noindent The weakening in case 2 of the preceding proposition is necessary.

\begin{example}[Nonlinear Linear-Time Dynamical System]
  Set $T = \{0\}$, giving rise to the singleton monoid which is
  trivially linear.  This fixes $\Phi$ completely as $\Phi(s, t) =
  \Phi(s, 0) = s$, giving rise to a ``still-life'' structure of time.
  Then neither $x \leadsto_{\mathbb S} y$ nor $y \leadsto_{\mathbb S}
  x$ for $x \neq y$.
\end{example}

\begin{definition}[Orbital Frame]
  A Kripke frame is called \emph{orbital} if and only if it
  corresponds, in the sense of
  Lemma~\ref{lemma:coalg-kripke-correspondence}, to the orbital
  coalgebra of some dynamical system.  An orbital frame is called
  linear-time/invertible if and only if it corresponds to the orbital
  coalgebra of some linear-time/invertible dynamical system,
  respectively.
\end{definition}

\noindent Using the preceding definition, Lemma~\ref{lemma:dyn2rel}
extends to Kripke frames.

\begin{lemma}
  \label{lemma:dyn2rel2}%
  For any orbital Kripke frame $\mathbb F = (W, R)$, the relation $R$ is
  \begin{enumerate}
  \item always a preorder,
  \item additionally non-branching if $\mathbb F$ is linear-time,
  \item additionally symmetric if $\mathbb F$ is invertible.
  \end{enumerate}
\end{lemma}

\noindent This statement has a partial, finitary converse.

\begin{lemma}
  \label{lemma:rel2dyn}%
  A finite Kripke frame $(W, R)$ is
  \begin{enumerate}
  \item always orbital if $R$ is a preorder,
  \item additionally linear-time if $R$ is non-branching,
  \item additionally invertible if $R$ is symmetric.
  \end{enumerate}
\end{lemma}

\begin{proof}
  Construct a dynamical system $\mathbb S = (\mathbb T, S, \Phi)$ with
  $(\leadsto_{\mathbb S}) = R$.  In any case, clearly $S = W$.
  Proceed in reverse order and increasing flexibility of cases.  For
  the latter two, consider the partition of $W$ into \emph{strongly
    connected components} (sccs) of the preorder $R$: maximal subsets
  $C \subseteq X$ such that $x \mathrel R y$ for all $x, y \in C$.  We
  write $x \sim y$ if and only if $x, y$ are in the same scc, that is
  $x \mathrel R y$ and $y \mathrel R x$.
  \begin{enumerate}
    \setcounter{enumi}{2}
  \item Set $\mathbb T = (\mathbb Z, 0, +)$.  By symmetry of $R$ there
    are no related pairs across sccs.  For each component $C$ choose
    an arbitrary cyclic permutation.  Set $\Phi^1$ to their
    composition.
    Then
    \begin{itemize}
    \item $x \leadsto_{\mathbb S} y$ via some $i < k$, where $k$ is the
      size of the scc containing both, if $x \mathrel R y$, and
    \item otherwise $x \not\leadsto_{\mathbb S} y$.
    \end{itemize}
    \setcounter{enumi}{1}%
  \item Set $\mathbb T = (\mathbb N, 0, +)$.  We say that $y$ is a
    \emph{successor} of $x$, writing $x \ll y$, if and only if $x
    \mathrel R y$ but not $y \mathrel R x$.  Clearly, $x \mathrel R y$
    if and only if either $x \sim y$ or $x \ll y$.  We say that $x$ is
    \emph{transient} if it has successors.  Since $W$ is finite and
    $R$ is non-branching, every transient $x$ has a unique least
    successor $x'$, and all elements reachable from $x$ are
    successors.  Set $\Phi^1(x) = x'$.  For non-transient $x$, all
    elements reachable from $x$ are in the same scc.  Proceed as
    above.  Then
    \begin{itemize}
    \item $x \leadsto_{\mathbb S} y$ via some $i < k$, where $k$ is the
      number of successors of $x$, if $x \ll y$,
    \item $x \leadsto_{\mathbb S} y$ via some $i < k$, where $k$ is the
      size of the scc containing both, if $x \sim y$, and
    \item otherwise $x \not\leadsto_{\mathbb S} y$.
    \end{itemize}
    \setcounter{enumi}{0}%
  \item There are in general no least successors, and there may
    non-successors reachable from transient elements.  A more basic
    construction is needed: Set $\mathbb T = (\mathbb N^*,
    \varepsilon, \cdot)$, the free monoid over $\mathbb N$.  For each
    $x \in W$ choose some infinite sequence $y = (y_0, y_1, \dots) \in
    W^\omega$ such that $x \mathrel R z$ if and only if $z = y_i$ for
    some $i$.  This is always possible since the set $\{ z \mid x
    \mathrel R z\}$ is finite and nonempty.  For the generating steps
    $\{ \Phi^n \mid n \in \mathbb N\}$, set $\Phi^n(x) = y_n$.  Then
    \begin{itemize}
    \item $x \leadsto_{\mathbb S} y$ via $1$, if $x \mathrel R y$, and
    \item otherwise $x \not\leadsto_{\mathbb S} y$.  \qedhere
    \end{itemize}
  \end{enumerate}
\end{proof}


\begin{theorem}
  The modal logics $S4$/$S4.3$/$S5$ are sound and complete for
  arbitrary/linear-time/invertible orbital frames, respectively.
\end{theorem}

\begin{proof}
  $S4$/$S4.3$/$S5$ are sound for the class of Kripke frames $(W, R$)
  where $R$ is an arbitrary/non-branching/symmetric preorder,
  respectively.  By Lemma~\ref{lemma:dyn2rel2}, they are also sound for
  the subclasses of arbitrary/linear-time/invertible orbital frames,
  respectively.

  $S4$/$4.3$/$S5$ are complete for the class of Kripke frames $(W, R$)
  where $R$ is an arbitrary/non-branching/symmetric preorder,
  respectively, and have the finite model property.  By
  Lemma~\ref{lemma:rel2dyn}, they are also complete for the
  subclasses of arbitrary/linear-time/invertible orbital frames,
  respectively.
\end{proof}

\begin{example}
  The operators $\Box$ and $\lozenge$ are well-suited to express
  ``long-term'' behavioral properties of dynamical systems.  For
  instance, let $A$ be the characteristic formula of a subset $U
  \subseteq S$ of the state space.  Then $U$ is a stationary solution
  of a dynamical system if and only if $A \to \Box A$ is valid in the
  Moss model associated with its orbits.
\end{example}

\section{Conclusion}

Many operators discussed in the temporal logic literature can be
subsumed under a common framework by viewing them as instances of
Moss's modality $\nabla$, for some coalgebraic presentation of the
underlying dynamical system models.  As a rule of thumb,
\begin{itemize}
\item step coalgebras go with discrete time,
\item trajectory coalgebras go with quantitative operators for either
  discrete or dense time, and
\item orbit coalgebras go with arbitrary time and qualitative
  operators, in particular the classical modal operators and the
  framework of normal modal logics.
\end{itemize}

The examples given in this article are of course only a small
selection to prove the viability of the approach.  There is
considerable potential for generalization.  The trajectory modality is
an extremely expressive tool, and it is likely that many other
temporal operators can be shown to coincide with particular
intensional notations for it.  Besides, coalgebraic perspectives on
dynamical systems other than the three detailed above could be
considered.  An interesting open problem and direction for future
research is the integration of measure-theoretic temporal operators,
for instance in duration calculus \cite{Chaochen1991}, into the
framework.



\appendix

\newcommand\doi[1]{\href{http://dx.doi.org/#1}{#1}}


\begin{thebibliography}{99}
\bibitem[1]{Birkhoff1927} G.~D.~Birkhoff. {\em Dynamical Systems.}
  American MAthematical Society, 1927.
\bibitem[2]{Bonsangue2005} M.~M.~Bonsangue and A.~Kurz. ``Duality for
  Logics of Transition Systems''.  In V.~Sassone (Ed.): {\em FoSSaS}.
  Lecture Notes in Computer Science 3441.  Springer, 2005,
  pp.~455--469.  \textsc{doi:} \doi{10.1007/978-3-540-31982-5\_29}.
\bibitem[3]{Chaochen1991} Z.~Chaochen, C.~A.~R.~Hoare and A.~P.~Raven.
  ``A Calculus of Durations''.  In: {\em Information Processing
    Letters} 40.5 (1991), pp.~269--276.
\bibitem[4]{Kurz2008} C.~C{\^\i}rstea et al.  ``Modal Logics are
  Coalgebraic''. In E.~Gelenbe, S.~Abramsy and V.~Sassone (Eds.): {\em
    BCS Int.\ Acad. Conf.}. British Computer Society, 2008,
  pp.~128--140.
\bibitem[5]{Jacobs2000} B.~Jacobs. ``Object-oriented kybrid systems of
  coalgebras plus monoid actions''. In: {\em Theoretical Computer
    Science} 239.1 (2000), pp.~41--95. \textsc{doi:}
  \doi{10.1016/S0304-3975(99)00213-3}.
\bibitem[6]{Moss1999} L.~S.~Moss. ``Coalgebraic Logic''.  In: {\em
    Ann.\ Pure Appl.\ Logic} 96.1--3 (1999), pp.~277--313.
  \textsc{doi:} \doi{10.1016/S0168-0072(98)00042-6}.
\bibitem[7]{Pratt1976} V.~Pratt.  ``Semantical Considerations on
  Floyd--Hoare Logic''.  In: {\em Proc.\ 17th Annual IEEE Symposium on
    Foundations of Computer Science}.  IEEE Computer Society, 1976,
  pp.~109--121.  \textsc{doi:} \doi{10.1109/SFCS.1976.27}.
\bibitem[8]{Rutten2000} J.~Rutten: ``Universal coalgebra: a theory of
  systems''.  In: {\em Theoretial Computer Science} 249.1 (2000),
  pp.~3--80.  \textsc{doi:} \doi{10.1016/S0304-3975(00)00056-6}.
\bibitem[9]{Venema2001} Y.~Venema.  ``Temporal Logic''. In L.~Goble
  (Ed.): {\em The Blackwell Guide to Philosophical Logic}.  Blackwell,
  2001, Chap.~10.  \textsc{doi:}
  \doi{10.1111/b.9780631206934.2001.00013.x}.
\end{thebibliography}
\end{document}